\documentclass[a4paper,UKenglish,cleveref, autoref, thm-restate]{lipics-v2021}

\hideLIPIcs
\nolinenumbers

\graphicspath{{./Figures/}}

\usepackage{amsmath,amssymb,amsthm,comment}
\usepackage{xspace}

\def\F{Fr\'echet\xspace}

\def\Rtwo{\mathbb{R}^{2}}
\def\M{\mathcal{M}}
\def\S{\mathcal{S}}
\def\eps{\varepsilon}

\usepackage{color}
\usepackage{graphicx}

\def\RR{\mathbb{R}}
\def\FF{\mathcal{F}}
\def\GG{\mathcal{G}}
\def\AA{\mathcal{A}}
\DeclareMathOperator*{\supp}{supp}
\DeclareMathOperator{\cost}{cost}
\DeclareMathOperator{\opt}{opt}
\DeclareMathOperator{\range}{range_\ge}

\title{On $k$-means for segments and polylines}

\author{Sergio Cabello}{Faculty of Mathematics and Physics, University of Ljubljana, Slovenia \and Institute of Mathematics, Physics and Mechanics, Slovenia}{sergio.cabello@fmf.uni-lj.si}{https://orcid.org/0000-0002-3183-4126}{Research partially supported by the Slovenian Research Agency (P1-0297, J1-2452, N1-0218, N1-0285)}

\author{Panos Giannopoulos}{Department of Computer Science, City, University of London, UK}{Panos.Giannopoulos@city.ac.uk}{https://orcid.org/0000-0002-6261-1961}{}

\authorrunning{S. Cabello and P. Giannopoulos}

\Copyright{Sergio Cabello and Panos Giannopoulos}



\begin{document}

\maketitle

\begin{abstract}
We study the problem of $k$-means clustering in the space of straight-line segments in $\Rtwo$ under the Hausdorff distance. For this problem, we give a $(1+\epsilon)$-approximation algorithm that, for an input of $n$ segments, for any fixed $k$, and with constant success probability, runs in time 
$O(n+ \eps^{-O(k)} + \eps^{-O(k)}\cdot \log^{O(k)} (\eps^{-1}))$.
The algorithm has two main ingredients. Firstly, we express the $k$-means objective in our metric space as a sum of algebraic functions and use the optimization technique of Vigneron~\cite{Vigneron14} to approximate its minimum. Secondly, we reduce the input size by computing a small size coreset using the sensitivity-based sampling framework by Feldman and Langberg~\cite{Feldman11, Feldman2020}. Our results can be extended to polylines of constant complexity with a running time of $O(n+ \eps^{-O(k)})$.
\end{abstract}


\section{Introduction}
The $k$-\emph{means} clustering problem is as follows: Given a point set in a metric space, find a set of points, called \emph{centers}, such that the sum of the squared distances from each input point to its closest center is minimized (over all possible choices of centers). It is a fundamental algorithmic problem with a ubiquitous role in data analysis in numerous application domains. As such, it has been studied extensively in geometric and general metric spaces, under various constraints on the objective and the choice of centers, and with a focus on complexity lower and upper bounds and the quality of the (approximate) solution~\cite{Ahmadian20, Awasthi15, Bandyapadhyay16, Chakrabarty22, Chen09, Cohen-Addad18,  Cohen-AddadSTOC22, Cohen-AddadFS21, Cohen-AddadG0LL19, Cohen-AddadSODA21, Cohen-AddadSL22, cohen-addad2019, Cohen-AddadSS21, Feldman11, Feldman2020, GrandoniORSV22, Har-Peled04, KumarSS10, Megiddo84}.

In geometric settings, almost all previous work involves clustering points in some low- or high-dimensional Euclidean space. A notable exception is the relatively recent work on $k$-\emph{center}~\cite{BuchinDGHKLS19} and $k$-\emph{median}~\cite{BuchinDR21, ChengH23, DriemelKS16, NathT22} clustering for polygonal curves; for $k$-center, one seeks to minimize the maximum distance to the closest center, while for $k$-median, one seeks to minimize just the sum of the distances (instead of the sum of the squares) to the closest centers. We are not aware of any work on $k$-means clustering for more complex objects than points.
In this paper, we consider the $k$-means problem in the spaces of segments and of polylines of constant complexity in the plane with respect to the Hausdorff distance.

\subsection{Formalization of the problem}

Let $(\S, d_H)$ be the metric space of closed straight-line segments in $\Rtwo$, where $d_H$ is the Hausdorff distance. Given a set $S$ of $n$ weighted segments, where each $s\in S$ has an associated positive weight $w_s \in \RR_{>0}$, and for any $k$ segments $s_1,\dots,s_k$ playing the role of ``centers'' of the clusters, 
we define the objective function
\[
	\cost_S(\{s_1,\dots,s_k\}) ~:=~ \sum_{s\in S} w_s \cdot \min \{ d_H^2(s_1,s), \dots, d_H^2(s_k,s)\} 
\]
and define the $k$-means problem as the problem of finding a minimizer, i.e., a set of segments $S^*=\{ s^*_1,\dots,s^*_k \}$ that minimizes the above cost. Note here that quite often we deal with \emph{unweighted} input segments. However, for technical reasons (made clear later in our discussion) in order to incorporate coresets in our algorithm, we have to consider the more general case of weighted segments.
Also note that we study the \emph{continuous} version of the problem, where the solution segments can come from anywhere in $(\S, d_H)$.  
This is harder than the so-called \emph{discrete} version, where the solution segments have to be selected among the input segments.

We also consider the $k$-means problem for polylines, each with a bounded number of segments, under the Hausdorff distance, where the definition of the problem is analogous.

We remark here on an interesting connection to the older and closely related concept of the \emph{\F mean}~\cite{F48}. This is a generalization of the classic notion of mean or average to any abstract metric space. For a finite point set $P$ in a metric space $(\M, d)$, a \F mean is any minimizer of the so-called \emph{\F variance} 
$\cost_P(q) ~:=~ \sum_{p\in P} d^2(q,p)$, taken over all $q\in \M$. 
For Euclidean spaces, the \F mean is the usual arithmetic mean. (Other usual means can be recovered as \F means by considering other distances.) 
The \F mean is a well-studied concept in Statistics and in Riemannian spaces, where sometimes it is known as Karcher mean, see~\cite{Schotz21} for a general, comprehensive treatment.
Computing a \F-mean is precisely the $1$-means clustering problem while the $k$-means is the generalization where the cost of each cluster is given by the functional defining the \F mean.

\subsection{Results}

Our main result is a $(1+\eps)$-approximation algorithm for the $k$-means problem in $(\S, d_H)$. The algorithm runs in $O\left(\left(n+ \eps^{-16k+4-\eta} + \eps^{-12k-3} \log^{4k+1} (\eps^{-1})\right) (\log(1/\delta)\right)$ time, for any fixed $k$, any $\eta>0$, and with success probability at least $1-\delta$ (the constant hidden in the $O$-notation depends on $\eta$ and $k$). 

There are two main ingredients in our algorithm. For the first one, described in Section~\ref{sec:algebraic}, we express the $k$-means objective in the space 
$(\S, d_H)$ as a sum of algebraic functions of constant description complexity. This algebraic approach allows us to use the optimization technique of Vigneron~\cite{Vigneron14} for approximating the minimum. This is, to the best of our knowledge, the first application of this technique in the context of clustering. While this technique is very expensive when applied directly to the entire set of input segments, we can decrease the running time dramatically by combining it with coresets. This is the second ingredient of our algorithm, described in Section~\ref{sec:coreset}, namely, we use the sensitivity framework of Feldman and Langberg~\cite{Feldman11, Feldman2020} to compute a small coreset of the input and then we apply the former algebraic approach to the coreset.

We then extend this result to polylines of description complexity at most $\ell=O(1)$. In this context, each input polyline and each solution polyline has at most $\ell$ segments, but we may put in the solution polylines that are not part of the input. The running time becomes $O\left(\left(n+ \eps^{-O(k\ell)}\right)\log(1/\delta)\right)$.

We start with a side-result, given in Section~\ref{sec:example}, where we consider the \F mean (or $1$-means) problem in a concrete example with two perpendicular segments that intersect at their centers. We show that even in this simple setting the set of \F means is surprisingly complex. The optimum is attained in a 3-dimensional subset of the 4-dimensional parameter space needed to model the space of candidate segments. This example also prompts to the benefit of looking into an algebraic approach for the general setting.

\subsection{Related work}

For general metric spaces, $k$-means (as well as $k$-median) clustering is APX-hard (when $k$ is part of the input)~\cite{Cohen-AddadSODA21, GuhaK99}. Several polynomial-time, constant factor approximation algorithms are known for both the continuous and discrete versions of the problem~\cite{Ahmadian20, Chakrabarty22}. For the discrete version, there even exist algorithms that achieve factors arbitrarily close to the lower bound~\cite{GuhaK99} and run in FPT-time with respect to $k$ and the approximation error $\eps$~\cite{cohen-addad2019}.

The Euclidean $k$-means, where the input is a set of points in $\RR^d$, is NP-hard for $d=2$~\cite{Megiddo84} and APX-hard when $d = \omega(\log n)$~\cite{Awasthi15}. The problem admits EPTASs with respect to $k$ and $\eps$ ~\cite{KumarSS10} and with respect to $d$ and $\eps$~\cite{Cohen-Addad18, Cohen-AddadFS21}.
 
As for the \F mean, it has been considered for persistence diagrams~\cite{Mileyko11,TurnerMMH14}, 
point sets on the unit circle~\cite{CazalsDO21}, and in the space of graphs~\cite{FergusonM22,Kolaczyk2020,Meyer21}, to name a few metric spaces far from the Euclidean setting.


\subsection{Definitions and notation}

For each point $p\in \Rtwo$, we use $x(p)$ and $y(p)$ for its two coordinates. Thus, $p=(x(p),y(p))$.
For any two points $p,q\in \RR^2$, we denote by $pq$ the segment with endpoints $p$ and $q$, and by $|pq|$
the Euclidean distance between them: $|pq|^2=(x(p)-x(q))^2+(y(p)-y(q))^2$. For simplicity we assume that all input segments have positive length.

Recall that the Hausdorff distance $d_H(A,B)$ between any two closed subsets $A,B\subset \RR^2$ is defined by
\[
	d_H(A,B) = \max \Bigl\{ \max_{a\in A}\min_{b\in B} |ab| ,\, \max_{b\in B}\min_{a\in A} |ab| \Bigr\}.
\]
Define $\delta(a, B) = \min_{q\in B} |ab|$ for the (directed) distance from a point $a$ to a closed set $B$.
It is well known and easy to see that for any two segments  $s_1=a_1b_1$ and $s_2=a_2b_2$ in $\S$
\begin{equation}
\label{eq:Hausdorff_segm}
d_H(s_1, s_2) = \max\{\delta(a_1, s_2), \delta(b_1, s_2), \delta(a_2, s_1), \delta(b_2, s_1)\}.	
\end{equation}


\section{An example of \F mean in $(\S, d_H)$}
\label{sec:example}

Let $s_1 = a_1 b_1$, $s_2 = a_2 b_2$ be perpendicular segments, centered at the origin $o$, with $|s_1| = |s_2| = 2$; see Figure~\ref{fig:2Seg_not_unique}. We show that the set of \F means in this case
is 3-dimensional. Since a candidate \F mean is a segment, described by four parameters, this example shows that even in relatively simple cases the space of \F means may have a complex structure and the maximum is attained at a large subset that is described algebraically. 

Let $D$ be the disk centered at $o$ with radius $1$ and let $D_1,\dots,D_4$ be the disks centered at $a_1/2,a_2/2,b_1/2,b_2/2$
with radius $\tfrac 12$. The region $D\setminus(D_1\cup\dots \cup D_4)$ has four connected
regions. Let $A_i$ be the closure of the connected region in the $i$-th quadrant; 
see Figure~\ref{fig:2Seg_not_unique}.

\begin{theorem}\label{thm:non_unique}
A segment is a \F mean of $\{ s_1, s_2 \}$ if and only if it goes through the origin and has its endpoints
in different regions $A_1,\dots,A_4$.
\end{theorem} 

Before proving the theorem, we show the following technical statement.

\begin{lemma}
\label{le:abab}
Consider any two segments $ab$ and $a'b'$ in the plane. 
Let $a''b''$ be the segment obtained by translating $a'b'$ such that its center coincides with the center of $ab$.
Then $d_H(ab,a''b'')\le d_H(ab,a'b')$.
Furthermore, if the center of $ab$ does not lie on $a'b'$ and the center of $a'b'$ does not lie on $ab$,
then $d_H(ab,a''b'')< d_H(ab,a'b')$.
\end{lemma}

\begin{proof}

\begin{figure}
\centering
	\includegraphics[width=11cm]{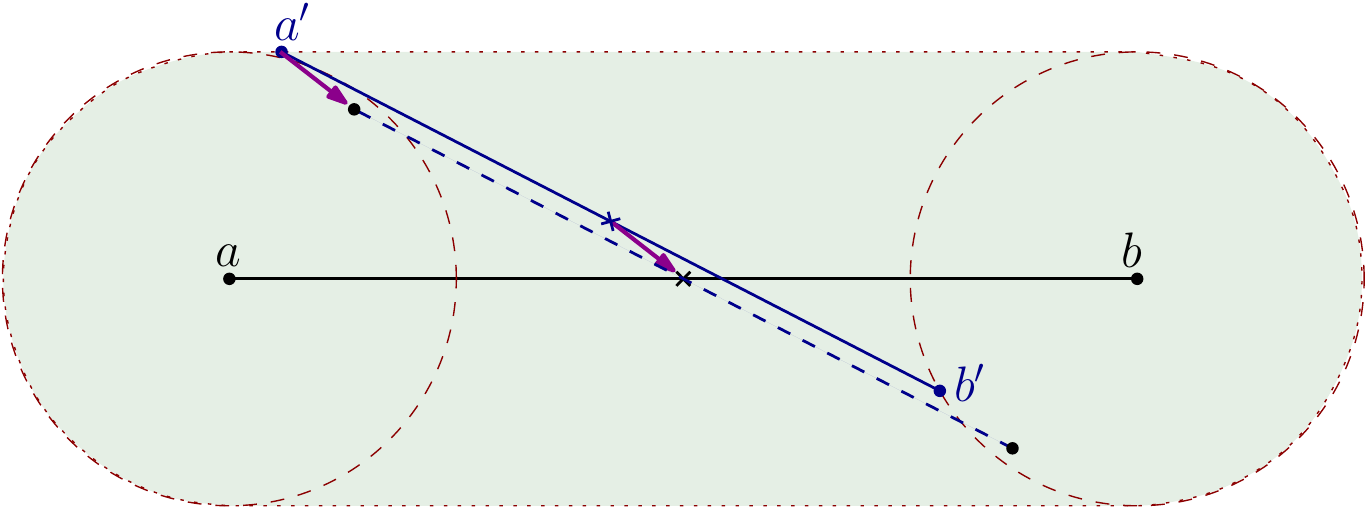}
	\caption{Translating $a'b'$ so that the centers of $ab$ and $a'b'$ coincide. The shaded region contains the points of the plane at distance at most $\delta(a',ab)$ from $ab$.}
	\label{fig:abab}
\end{figure}

Without loss of generality we may assume that $\delta(a',ab)\ge \delta(b',ab)$. 
See Figure~\ref{fig:abab}.
Note that 
\[
	\delta(a'',ab) \le \delta(a',ab),
\]
and equality can occur only when the center of $a'b'$ lies on $ab$.
Because of symmetry we have $\delta(a'',ab)=\delta(b'',ab)$.
Therefore 
\[
	\max\{ \delta(a'',ab), \delta(b'',ab)\} = \delta(a'',ab)
	\le \delta(a',ab) = \max\{ \delta(a',ab), \delta(b',ab)\},
\]
and equality can occur only when the center of $a'b'$ lies on $ab$.
Exchanging the roles of $ab$ and $a'b'$, we also have 
\[
	\max\{ \delta(a,a''b''), \delta(b,a''b'')\} \le \max\{ \delta(a,a'b'), \delta(b,a'b')\},
\]
and equality can occur only when the center of $ab$ lies on $a'b'$.
We then have
\begin{align*}
	d_H(ab,a'b') &= \max\{ \delta(a,a'b'),\delta(b,a'b'),\delta(a',ab),\delta(b',ab)\}\\
	&\geq \max\{ \delta(a,a''b''),\delta(b,a''b''),\delta(a'',ab),\delta(b'',ab)\} \\
	&= d_H(ab,a''b'').
\end{align*}
Moreover, when the center of $ab$ does not lie on $a'b'$ and the center of $a'b'$ does not lie on $ab$,
the inequality is strict. 
\end{proof}

\begin{proof}[Proof of Theorem~\ref{thm:non_unique}]

\begin{figure}
\centering
	\includegraphics[width=7.5cm]{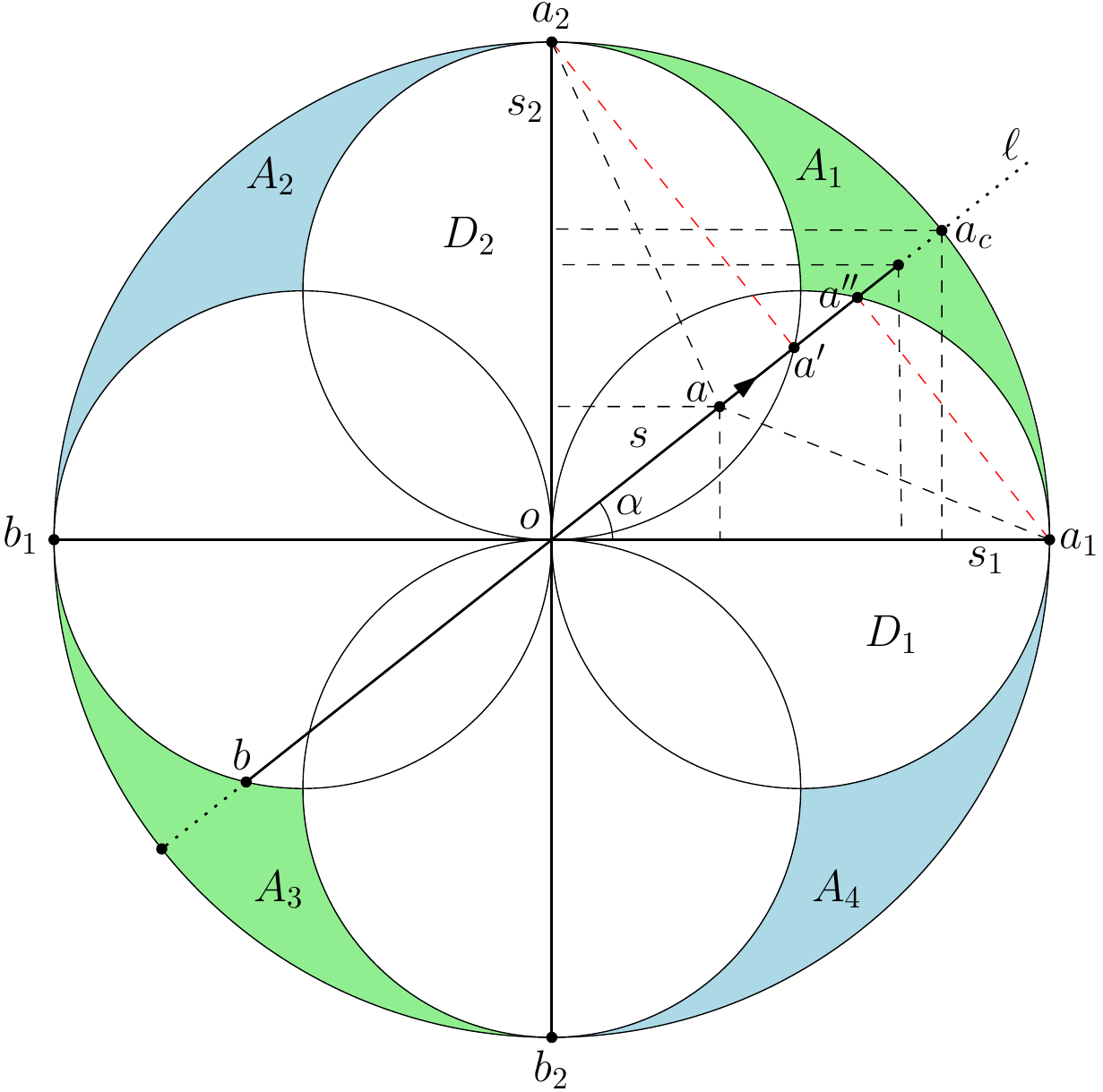}
	\caption{For the set of segments $S = \{ s_1, s_2 \}$, there is an infinite family of \F means.}
	\label{fig:2Seg_not_unique}
\end{figure}

Note that $d_H(s_1, s_2) = 1$. Therefore, for any solution (\F mean) segment $s$ we have $d^2_H(s, s_1) + d^2_H(s, s_2) \leq d^2_H(s_1, s_1) + d^2_H(s_1, s_2) = 0 + d^2_H(s_1, s_2) = 1$.

Observe that any solution segment $s$ must pass through the origin. 
Indeed, if such a solution $s$ does not pass through the origin, then the center of $s$
does not lie on $s_1$ or it does not lie on $s_2$.
It follows from Lemma~\ref{le:abab} that, by translating $s$ such that its center coincides with the origin, both $d_H(s, s_1)$ and $d_H(s, s_2)$ do not increase,
and at least one of them strictly decreases. Therefore, $d^2_H(s, s_1)+ d^2_H(s, s_2)$ strictly decreases by translating $s$ such that its center coincides with the origin.

Next, consider a solution segment $s = ab$ that passes through the origin $o$ with its endpoints $a, b$ in the first and third quadrant, respectively. See Figure~\ref{fig:2Seg_not_unique}. Let $a_c$ and $a'$ be the points of intersection of the supporting line $\ell$ of $s$ with the boundary of $D$ and the boundary of $D_2$ respectively.
Note that $\angle o a' a_2$ is right.
As $a$ moves from $o$ to $a'$, we have $\delta(a_2, s) = |a_2 a| > \delta(a, s_2)$ and the unique minimum of $|a_2 a|$ is attained at $a=a'$. 
When $a$ moves from $a'$ to $a_c$, $\delta(a_2, s)=|a_2a'|$ is constant. At $a=a_c$ we have that $\delta(a_2, s) = \delta(a_c, s_2)$, and beyond that $\delta(a_2, s) < \delta(a, s_2)$. 
Similarly, the minimum value of $\delta(a_1, s)$ is $|a_1a''|$, where $a''$ is the projection of $a_1$ onto $\ell$ (the intersection of $\ell$ with the boundary of $D_1$). We also have that $\delta(a_1, s) = \delta(a_c, s_1)$ when $a$ lies between $a''$ and $a_c$, and beyond $a_c$, $\delta(a_1, s) < \delta (a, s_1)$. 

Assume that $a''\in a'a_c$, as in Figure~\ref{fig:2Seg_not_unique}; the other case is symmetric.
Using $\alpha:=\angle aoa_1$, for every position of $a$ in the segment $a''a_c$ we have 
\[
\delta^2(a, s_1) +  \delta^2(a, s_2) ~\le~ \delta^2(a_1, s) +  \delta^2(a_2, s) ~=~ |a_1a''|^2 + |a_2a'|^2 ~=~ \sin^2\alpha + \sin^2(\pi/2 - \alpha) ~=~ 1.
\] 
For every position of $a$ in the segment $oa''$ we have 
$\delta^2(a_1, s) +  \delta^2(a_2, s) ~>~ |a_1a''|^2 + |a_2a'|^2 ~=~  1$,
while for every position of $a$ past $a_c$ we have that
$\delta^2(a, s_1) +  \delta^2(a, s_2) > \delta^2(a_1, s) + \delta^2(a_2, s) = 1$.

For the endpoint $b$ of $s$ in the third quadrant, the situation is symmetric. Using that the Hausdorff distance is attained at some endpoint, 
we get that each segment with an endpoint in $A_1$ and an endpoint in $A_3$ gives the minimum possible value of $d^2_H(s, s_1) + d^2_H(s, s_2)\ge 1$, and therefore is a \F mean. The situation is symmetric for the other quadrants.
\end{proof}

\section{An algebraic approach to $k$-means in $(\S, d_H)$}
\label{sec:algebraic}

We use the following adaptation of the definition of a \emph{nice} family of functions by Vigneron~\cite[Section 2.1]{Vigneron14}.
Let $\FF=\{ f_i:\RR^d \rightarrow \RR \mid i\in I\}$ be a finite family of functions, where $I$ is some index set.
We say that $\FF$ is \emph{nice} if there exists a constant $\lambda>d>0$ such that:
\begin{itemize}
\item each $f_i\in \FF$ is nonnegative and bounded;
\item for each $f_i\in \FF$, there exists a semialgebraic set $\supp(f_i)\subseteq \RR^d$
	and an algebraic function $g_i$ of degree at most $\lambda$ with
	$f_i(x)=g_i(x)$ for $x\in \supp(f_i)$ and $f_i(x)=0$ for $x\notin \supp(f_i)$;
\item for each $f_i\in \FF$, the semialgebraic set $\supp(f_i)\subseteq \RR^d$ is  
	a boolean combination of at most $\lambda$ subsets of $\RR^d$,
	each of them defined by an polynomial inequality of degree at most $\lambda$;
\item for each $f_i\in \FF$, the restriction of $f_i$ to $\supp(f_i)$ is continuous.
\end{itemize}
Note that the definition allows that the sets $\supp(f_i)$ are open, closed or mixed.
It also allows that $f_i$ is discontinuous in $\RR^d\setminus\supp(f_i)$, which
may include the boundary of $\supp(f_i)$ in some cases.

Our use of this concept will be through the following result for computing an approximation to the minimum of 
the function $\sum_i f_i$.

\begin{theorem}[Adaptation of Theorem 3.4 in Vigneron~\cite{Vigneron14}]
\label{thm:Vigneron}
	Assume that $\eps \in (0,1)$.
	Let $\FF=\{ f_i:\RR^d \rightarrow \RR \mid i\in I\}$ be a nice family of $m$ functions.
	Define $g=\sum_{i\in I} f_i$ and assume that $\min_{x\in \RR^d} g(x)$ exists. 
	Then we can compute a point $x'_\eps \in \RR^d$ such that
	$g(x'_\eps)\le (1+\eps) \min_{x\in \RR^d} g(x)$ 
	in time $O(m^{2d-2+\eta} + (m/\eps)^{d+1} \log^{d+1}(m/\eps))$
	for any $\eta>0$.\footnote{As noted in Vigneron~\cite{Vigneron14}, one needs to assume either the Real-RAM model of computation (which is standard in computational geometry) or a model where we can choose the precision of the intermediate computations, and then the computational complexity of the algorithm depends on the desired precision.} 
	The constant hidden in the $O$-notation depends on $\eta$ and on $d$.	
\end{theorem}

Let us first consider the simpler case of two segments and how their Hausdorff 
distance is defined.
We parameterize a segment $ab$ as the point 
$(x(a),y(a),x(b),y(b))$ in $\RR^4$.
Note that in this parameterization we have the artifact that
the segments $ab$ and $ba$ give different points in $\RR^4$.

Let $\ell_{ab}$ be the line supporting a segment $ab$.
For a point $p$, the distance $\delta(p,ab)$ is given by one of the three terms $|pa|$, $|pb|$, or $\delta(p,\ell_{ab})$.
For a point $q\in \Rtwo$ and
a segment $ab$, let $\ell_\perp(q,ab)$ be the line perpendicular to $\ell_{ab}$ through $q$.
The lines $\ell_\perp(a,ab)$ and $\ell_\perp(b,ab)$ 
partition the plane into three 2-dimensional faces (Figure~\ref{fig:algebra2}) with closures
\begin{align*}
	\sigma(ab)&= \text{the closed slab between $\ell_\perp(a,ab)$ and $\ell_\perp(b,ab)$},      \\
	\tau(a,ab)&= \text{the closed halfspace defined by $\ell_\perp(a,ab)$ that does not contain $b$}, \\
	\tau(b,ab)&= \text{the closed halfspace defined by $\ell_\perp(b,ab)$ that does not contain $a$}.
\end{align*}
We then have
\[
	\delta(p,ab) ~=~\begin{cases}
				|pa| &\text{if $p\in \tau(a,ab)$,}\\
				|pb| &\text{if $p\in \tau(b,ab)$,}\\
				\delta(p,\ell_{ab})&\text{if $p\in \sigma(ab)$}.
			\end{cases}
\]
From Equation~\eqref{eq:Hausdorff_segm}, we conclude that, for any two segments $ab$ and $a'b'$, the distance $d_H(ab,a'b')$ is given by one
of the functions in the family
\[
	\FF(ab,a'b')~:=~
		\bigl\{ |aa'|,~ |ab'|,~ |ba'|,~ |bb'|,~
			\delta(a,\ell_{a'b'}),~ \delta(b,\ell_{a'b'}),~ \delta(a',\ell_{ab}),~ \delta(b',\ell_{ab})\bigr\}. 
\]

\begin{figure}
\centering
	\includegraphics[page=2]{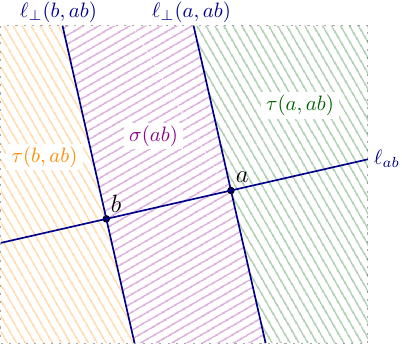}
	\caption{The regions $\sigma(ab)$, $\tau(a,ab)$ and $\tau(b,ab)$.}
	\label{fig:algebra2}
\end{figure}

We next argue that all the expressions involved are algebraic.
A point $p$ lies on the line $\ell_\perp(a,ab)$ if and only 
the scalar product of the vectors $\vec{ap}$ and $\vec{ab}$ is zero.
This is equivalent to $\bigl( x(p),y(p),x(a),y(a),x(b),y(b)\bigr)$ 
being a zero of the algebraic (actually polynomial) function
\[
	\psi(x,y,x_a,y_a,x_b,y_b) := (x-x_a)(x_b-x_a)+(y-y_a)(y_b-y_a).
\]
The sign of this expression also tells us which side of $\ell_\perp(a,ab)$ the point $p$ lies on.
Note that this function is linear in $x$ and $y$, while it is quadratic in $x_a$ and $y_a$.
Symmetrically, the sign of $\psi\bigl( x(p),y(p),x(b),y(b),x(a),y(a) \bigr)$ tells us
which side of $\ell_\perp(b,ab)$ point $p$ lies on.

In the following, we will treat the segment $a'b'$ as variable, identified with $\RR^4$,
while the segment $ab$ will be fixed.
We will show that the space $\RR^4$ can be decomposed into cells such that, within a cell, the distance $d_H(ab,a'b')$ is defined always by the same function from $\FF(ab,a'b')$. 
Such a decomposition is given by the eight algebraic hypersurfaces describing
the conditions
\begin{align*}
	&a'\in \ell_\perp(a,ab),~ b'\in \ell_\perp(a,ab),~ a'\in \ell_\perp(b,ab),~ b'\in \ell_\perp(b,ab), \\
	&a\in \ell_\perp(a',a'b'),~ b\in \ell_\perp(a',a'b'),~ a\in \ell_\perp(b',a'b'),~ b\in \ell_\perp(b',a'b'),
\end{align*}
together with a set of hypersurfaces, ``bisectors'', each defined by the set of points where two appropriate functions from $\FF(ab,a'b')$ meet; this will become clear shortly.
Finally, we note that each function in $\FF(ab,a'b')$ is algebraic of constant degree;
for example, elementary algebra shows that
\[
	\delta^2(a',\ell_{ab}) ~=~ 
		\frac{\Bigl( \bigl( x(b)-x(a)\bigr) \bigl( y(a)-y(a') \bigr)-\bigl( x(a)-x(a')\bigr) \bigl( y(b)-y(a) \bigr)\Bigr)^2}{\bigl( x(a)-x(b)\bigr)^2 + \bigl(y(a)-y(b)\bigr)^2} .
\] 

We parameterize the space of (sequences of) $k$ segments $a_1b_1,\dots, a_kb_k$ (the $k$ candidate cluster centers) by the point 
\[
	\bigl( x(a_1),y(a_1),x(b_1),y(b_1),\dots, x(a_k),y(a_k),x(b_k),y(b_k)\bigr)\in \RR^{4k}.
\]
Similarly, each $z\in \RR^{4k}$ defines a $k$-tuple of segments with
$s_1(z)=a_1(z)b_1(z),\dots, s_k(z)=a_k(z)b_k(z)$ by taking the inverse of the parameterization.

\begin{theorem}
\label{thm:nice_family}
	Let $k$ be a fixed, positive integer and let $s$ be a segment in the plane.
	In $O(1)$ time we can construct  
	a nice family $\FF_s=\{ f:\RR^{4k} \rightarrow \RR\}$ of $O(1)$ functions such that
	\[
		\forall z\in \RR^{4k}: ~~~ \sum_{f\in \FF_s} f(z) = \min_{i\in [k]} d^2_H(s,s_i(z)).
	\]
\end{theorem}
\begin{proof}
	Let $s=ab$ be the fixed segment.
	For each index $i\in [k]$, we consider the set $\Sigma(i)$
	of 8 hypersurfaces in $\RR^{4k}$, 
	each of them given by one of the following conditions
	\begin{align*}
		& a_i\in \ell_\perp(a,ab),~ b_i\in \ell_\perp(a,ab),~ a_i\in \ell_\perp(b,ab),~ b_i\in \ell_\perp(b,ab),\\
		& a\in\ell_\perp(a_i,a_ib_i),~ b\in\ell_\perp(a_i,a_ib_i),~ a\in\ell_\perp(b_i,a_ib_i),~ b\in\ell_\perp(b_i,a_ib_i).
	\end{align*}
	Note that here $x(a)$, $y(a)$, $x(b)$ and $y(b)$ are input data  
	while $x(a_i)$, $y(a_i)$, $x(b_i)$ and $y(b_i)$ are variables defining
	coordinates in the parameter space $\RR^{4k}$.

	Set $\Sigma := \cup_{i\in [k]} \Sigma(i)$ and let
	$\AA_\Sigma$ be the arrangement in $\RR^{4k}$ defined by $\Sigma$. 
	From the foregoing discussion, we have the following property:
	for each cell $c$ of $\AA_\Sigma$ and each index $i\in [k]$, the set of functions $\FF(s,a_ib_i)$ stays the same
	and each of the distances $\delta(a, a_ib_i), \delta(b, a_ib_i), \delta(a_i, ab),$ and $\delta(b_i, ab)$, 
	is given by the same function from $\FF(s,a_ib_i)$. Thus, for all $z\in c$, 
	the distance $d_H(s,s_i(z))$ is described by the maximum among the same four functions from $\FF(s,a_ib_i)$. 
	
	In order to make clear that only the coordinates of $a_i$ and $b_i$ are relevant
	in the functions in $\FF(s,a_ib_i)$, we change the notation to $\GG_i$ 
	and take each function $g$ of $\GG_i$ to map from $\RR^{4k}$ to $\RR$.
	Formally, for each function $f\in \FF(s,a_ib_i)$ we put 
	into $\GG_i$ the function $g(z):=f(ab,s_i(z))$.

	We next define a set $\Lambda$ of algebraic hypersurfaces in $\RR^{4k}$ playing the role of ``bisectors''. 
	For each $i,j\in [k]$ with $i \leq j$, we define $\Lambda(i,j)$ as the hypersurfaces
	given by equating each function of $\GG_i$ to each function of
	$\GG_j$. 
	Note that each hypersurface is defined by a polynomial equality of degree at most $6$.
	Since $\GG_i$ has $8$ functions for each $i\in [k]$, 
	the set $\Lambda(i,j)$ has at most $8^2=64$ hypersurfaces (it is $32$ for $\Lambda(i,i)$).

	Set $\Lambda :=\cup_{i\in [k]} \cup_{j\in [k], i\leq j} \Lambda(i,j)$
	and let $\AA_\Lambda$ be the arrangement in $\RR^{4k}$ induced by $\Lambda$.
	For each cell $c\in \AA_\Lambda$ the sign of each function
	$g(z)-g'(z)$ remains constant for $g\in \GG_i$, $g'\in \GG_j$ and $z\in c$. 

	Finally, let $\AA$ be the arrangement in $\RR^{4k}$ induced
	by the hypersurfaces in $\Sigma\cup \Lambda$. 
	Note that this is a refinement
	of $\AA_\Sigma$ and $\AA_\Lambda$, meaning that each cell of $\AA$ is contained 
	in a cell of $\AA_\Sigma$ and a cell of $\AA_\Lambda$.

	Consider a cell $c\in \AA$. 
	Since $c$ is contained in a cell of $\AA_\Sigma$, for each $i\in [k]$, each function in the set 
	$\Delta_i(c) = \{ \delta(a, s_i(z)), \delta(b, s_i(z)), \delta(a_i(z), ab),$ and $\delta(b_i(z), ab) \}$
	is given by the same function of $\GG_i$ for all $z\in c$. Moreover, since $c$ is contained in a cell of 
	$\AA_\Lambda$, for every two distinct functions $\delta, \delta' \in \Delta_i(c)$ the sign of $\delta - \delta'$ is 	
	constant for all $z\in c$. From these two facts we conclude that, for each $i\in [k]$, there is some function 	
	$g_{c,i}(z) \in \GG_i$ such that $d_H(s,s_i(z)) = g_{c,i}(z)$ for all $z\in c$. This function can be easily determined in 
	$O(1)$ time per cell by evaluating each function in $\Delta_i(c)$ at some arbitrary point in $c$.
	
	Similarly, since $c$ is contained in a cell of $\AA_\Lambda$, we have that 
	for each distinct $i,j\in [k]$ the sign of 
	\[
		d_H(s,s_i(z)) - d_H(s,s_j(z)) ~=~ g_{c,i}(z)- g_{c,j}(z)
	\]
	is constant for all $z\in c$. 
	This implies that, for each cell $c\in \AA$, 
	there exists some index $\iota(c)\in [k]$ with the following property:
	\begin{align*}
		&\forall j\in [k],~ z\in c:~~~ d_H(s,s_{\iota(c)}(z)) \le d_H(s,s_j(z)).
	\end{align*}
	In other words, the segment $s_{\iota(c)}(z)$ is a closest one to $s$ among $s_1(z),
	\dots,s_k(z)$
	and moreover the distance $d_H(s,s_{\iota(c)}(z))$ is given by a single
	function $g_{c,\iota(c)}$ from $\GG_{\iota(c)}$.
	Thus, for each $z\in c$ it holds 
	$\min_{i\in [k]} d_H(s,s_i(z))=g_{c,\iota(c)}(z)$. As before, the function $g_{c,\iota(c)}(z)$ 
	can be determined in $O(k)$ time per cell by evaluating each $d_H(s,s_i(z))$ 
	at some arbitrary point in $c$. 
	
	For any set $A$, let $1_A$ be the function with $1_A(x)=1$ if $x\in A$ and $1_A(x)=0$ if $x\notin A$.
	For each cell $c\in \AA$, define the function $h_c: \RR^{4k} \rightarrow \RR$ by $h_c(z)=1_c (z) \cdot g^2_{c,\iota(c)}(z)$.
	Finally, set $\FF_s := \{ h_c\mid c\in \AA \}$.
	We can then express the function 
	\[
		z\in \RR^{4k} \mapsto \min_{i\in [k]} d^2_H(s,s_i(z))
	\]
	as 
	\[
		\min_{i\in [k]} d^2_H(s,s_i(z)) ~=~ \sum_{c\in \AA} 1_{c}(z)\, g^2_{c,\iota(c)}(z) 
		~=~ \sum_{c\in \AA} h_c(z) ~=~ \sum_{h\in \FF_s} h(z).
	\]
	Since $\Sigma\cup\Lambda$ has $O(k^2)=O(1)$ hypersurfaces, the arrangement 
	$\AA$ has $O(O(k^2)^{4k}) = O(1)$ cells, each of them described by $O(k^2)=O(1)$
	algebraic inequalities of constant description complexity and the family of functions $\FF_s$ has the desired 
	properties, where the constant $\lambda$ used to define the niceness is $O(k^{8k})$. Constructing $\AA$ (i.e., with algebraic descriptions for each cell) takes $O(O(k^2)^{4k+1} 6^{O((4k)^4)}) = O(1)$~\cite[Chapter~16]{AlgRealAlgGeom06}. The family $\FF_s$ can be constructed in this time as well.
\end{proof}

We can now apply Theorem~\ref{thm:Vigneron} combining all the functions $\FF_s$ for $s\in S$ and compute a set of $k$ segments whose cost approximates that of an optimal set of segments. 

\begin{theorem}
\label{thm:main_via_Vigneron}
	Let $k$ a fixed, positive integer and let $\eps\in (0,1)$.
	Let $S$ be a family of $n$ segments in the plane with positive weights.
	We can compute $k$ segments $s_{1,\eps},\dots,s_{k,\eps}$ in $\RR^2$ such that
	\[
		\cost_S(\{s_{1,\eps},\dots,s_{k,\eps}\}) ~\le~ (1+\eps) \min \Bigl\{ \cost_S(\{s_1,\dots,s_k\}) \mid s_1,\dots,s_k\text{ segments}\Bigr\}.
	\]
	in time $O(n^{8k -2+\eta} + (n/\eps)^{4k+1} \log^{4k+1} (n/\eps))$, for any $\eta > 0$. 
	The constant hidden in the $O$-notation depends on $\eta$ and on $k$.
\end{theorem}
\begin{proof}
	For each segment $s\in S$ we compute the family $\FF_s$ of Theorem~\ref{thm:nice_family}. 
	This takes $O(n) \cdot O(1) = O(n)$ time in total. 	
	To account for the weight $w_s>0$ of the segment $s$, 
	we replace in $\FF_s$ each function $f\in \FF_s$ with $w_s\cdot f$.
	Define $\FF:=\cup_{s\in S} \FF_s$ and the function $g:=\sum_{f\in \FF} f$. 
	Note that $\FF$ is a family of 
	$O(1) \cdot O(n) = O(n)$ nice functions and
	\[
		\forall z\in \RR^{4k}: ~~~ g(z) ~=~ \sum_{s\in S}\sum_{f\in \FF_s} f(z)
													 ~=~ \sum_{s\in S} w_s\cdot \min_{i\in [k]} d_H(s,s_i(z))^2
													~=~ \cost_S(\{ s_1(z),\dots, s_k(z)\}).
	\]
	We can then use Theorem~\ref{thm:Vigneron} to find 
	in time $O(|\FF|^{2\cdot 4k -2+\eta} + (|\FF|/\eps)^{4k+1} \log^{4k+1} (|\FF|/\eps))$, for any $\eta>0$,
	a point $z'_\eps\in \RR^{4k}$ such that 
	\[
		g(z'_\eps) ~\le~ (1+\eps) \min_{z\in \RR^{4k}} \cost_S(\{ s_1(z),\dots, s_k(z)\}).
	\]
	The point $z'_\eps\in \RR^{4k}$ defines the segments 
	$s_{1,\eps}:=s_1(z'_\eps),\dots,s_{k,\eps}:=s_k(z'_\eps)$.	
	Since $s_1(z),\dots,s_k(z)$ goes over all $k$ tuples of segments when $z$ iterates over all $\RR^{4k}$,
	we have 
	\[
		\min_{z\in \RR^{4k}} \cost_S(\{ s_1(z),\dots, s_k(z)\}) ~=~ \min_{s_1,\dots,s_k} \cost_S(\{ s_1,\dots, s_k\}).
	\]	
	We conclude that 
	\[
		\cost_S(\{s_{1,\eps},\dots,s_{k,\eps}\} ~=~ g(z'_\eps) 
			~\le~ (1+\eps) \min_{s_1,\dots,s_k} \cost_S(\{ s_1,\dots, s_k\}).	\qedhere
	\]
\end{proof}


\section{A coreset for $k$-means in $(\S, d_H)$}
\label{sec:coreset}

We use the sensitivity framework of Feldman and Langberg~\cite{Feldman11, Feldman2020}.
Let $F$ be a finite set of functions, each of them mapping from $\RR^d$ to $\RR_{\ge 0}$.
The \emph{sensitivity} of $f\in F$ with respect to $F$ is 
\[
	\sigma(f,F) ~:=~ \sup_{z\in \RR^d} \frac{f(z)}{\displaystyle \sum_{g\in F} g(z)}.
\]
We also consider the following range space
\[
	\range(F) ~:=~ \left( F, \bigl\{ \{ f \in F \mid f(z)\ge r \} \mid z\in \RR^d,~ r \in [0,\infty) \bigr\} \right).
\]
We will use the following theorem from~\cite{Feldman2020},
which we state here adapted to our needs.

\begin{theorem}[Adaptation of Theorem 31 in Feldman et al.~\cite{Feldman2020}]
\label{CoresetThm}
	Let $F$ be a set of $n$ functions from $\RR^d$ to $[0,\infty)$ with the following properties:
	\begin{itemize}
	\item For each choice of  weights $w_f > 0$ for $f\in F$, the range space 
	$\range(\{w_f\cdot f \mid f\in F \})$ has bounded VC-dimension.
	\item For each $f\in F$ we are given a value $\tilde\sigma(f)$ such that
	\[ \tilde\sigma(f) \ge \frac{1}{|F|} ~~~\text{ and }~~~ \tilde\sigma(f) \ge \sigma(f,F).
	\]
	\end{itemize}
	Set $\tilde\Sigma(F):=\sum_{f\in F} \tilde\sigma(f)$.
	Let $\delta,\eps$ be real values in $(0,1/2)$.
	In time $O(|F|)$ we can compute a subset $C\subseteq F$ of 
	\[
		O\left( \frac{\tilde\Sigma(F)}{\eps^2}\left( \log \tilde\Sigma(F) + \log \frac{1}{\delta}\right) \right)
	\]
	weighted functions and weights $u_f> 0$ for each $f\in C$ such that, 
	with probability at least $1-\delta$:
	\[
		\forall z\in\RR^d:~~~ \left| \sum_{f\in F} f(z) - \sum_{f\in C}u_f\cdot f(z)\right| ~\le~
		\eps \sum_{f\in F} f(z).
	\]
\end{theorem}

For each input segment $s\in S$, we define the function $f_s:\RR^{4k}\rightarrow \RR_{\ge 0}$ with
\[
	f_s(z)~:=~ \min \{ d^2_H(s,s_1(z)),\dots, d^2_H(s,s_k(z))\} ~=~ \bigl( \min \{ d_H(s,s_1(z)),\dots, d_H(s,s_k(z))\}\bigr)^2.
\]
Here, the segments $s_1(z),\dots,s_k(z)$ are the same that were used in the parameterization before Theorem~\ref{thm:nice_family}.
Set $F=\{ f_s \mid s\in S\}$.
In order to use the above theorem, we need appropriate sensitivity upper bounds $\tilde\sigma(f_s)$ for each $f_s \in F$ and a bound on the total sensitivity $\tilde\Sigma(F)$. 
Let $\opt_k(S)$ be the cost of an optimal set of segments for $k$-means, i.e., $\opt_k(S) = \min_{s_1,\dots,s_k} \cost_S(\{ s_1,\dots, s_k\})$.

\begin{lemma}
\label{le:sensitivity}
	Let $s'_1\dots,s'_{k'}$ be a bicriteria $(\alpha,\beta)$-approximation for $k$-means, that is, 
	$k'\le \beta k$ and $\cost_S(\{s'_1,\dots,s'_{k'}\})\le \alpha\cdot \opt_k(S)$, where $\alpha,\beta\ge 1$.
	For each $i\in [k']$, let $S'_i$ be the segments of $S$ closer to $s'_i$ 
	than to any other segment $s'_j$, $j\in [k']\setminus \{i\}$; 
	ties are solved arbitrarily so that $S'_1,\dots, S'_{k'}$ is a partition of $S$.
	For each segment $s\in S$, let $\iota(s)\in [k']$ be such that $s\in S'_{\iota(s)}$.
	Define for each $s\in S$ the value
	\[
		\tilde\sigma(f_s)~:=~ \frac{32\alpha}{|S'_{\iota(s)}|} + \frac{16\alpha\cdot d^2_H(s,s'_{\iota(s)})}{\displaystyle\sum_{s'\in S'_{\iota(s)}} d^2_H(s',s'_{\iota(s)})} 
			~=~ \frac{32\alpha}{|S'_{\iota(s)}|} + \frac{16\alpha\cdot d^2_H(s,s'_{\iota(s)})}{\cost_{S'_{\iota(s)}}(s'_{\iota(s)})}. 
	\]
	Then $\tilde\sigma(f_s)\ge \sigma(f_s,F)$ and $\tilde\sigma(f_s)\ge \frac{1}{|F|}$.
\end{lemma}

\begin{proof}
	Consider a fixed $s\in S$ and set $i=\iota(s)$; thus $s\in S'_i$.
	Define $R_i$ to be the average cost of the elements of $S'_i$ for the center $s'_i$:
	\[
		R_i ~=~ \frac{\cost_{S'_i}(s'_i)}{|S'_i|} ~=~ \frac{\displaystyle\sum_{s'\in S'_i} d^2_H(s',s'_i)}{|S'_i|}.
	\]
	With this notation we have 
	\[
		\tilde\sigma(f_s) ~=~ \frac{32\alpha}{|S'_i|} + \frac{16\alpha\cdot d^2_H(s,s'_i)}{\cost_{S'_i}(s'_i)} 
				~=~ \frac{32\alpha}{|S'_i|} + \frac{16\alpha\cdot d^2_H(s,s'_i)}{|S'_i| \cdot R_i}. 
	\]
	Since $|S'_i|\le |S| =|F|$ and $\alpha\ge 1$, 
	we have $\tilde\sigma(f_s) \ge \frac{1}{|S'_i|} \ge \frac{1}{|F|}$.
	
	For the other inequality, recall that 
	\[
		\sigma(f_s,F) ~=~ \sup_{z\in \RR^{4k}} \frac{f_s(z)}{\displaystyle \sum_{f_{s'}\in F} f_{s'}(z)} 
			~=~ \sup_{z\in \RR^{4k}} \frac{\min\{ d^2_H (s,s_1(z)),\dots, d^2_H (s,s_k(z))\}}{\displaystyle \sum_{s'\in S} \min\{ d^2_H (s',s_1(z)),\dots, d^2_H (s',s_k(z))\} }.
	\]
	
	Consider any $z\in \RR^{4k}$. We distinguish three cases.
	
	\medskip
	\noindent
	First, if $z$ is such that for some $j\in [k]$ we have $d^2_H(s,s_j(z)) \le 16\cdot d^2_H(s,s'_i)$,
	then we use that $\cost_S(\{ s'_1,\dots,s'_{k'} \}) \le \alpha\cdot \opt_k(S)$ to obtain
	\[
		\frac{f_s(z)}{\displaystyle \sum_{s'\in S} f_{s'}(z)} ~\le~ 
			\frac{d^2_H(s,s_j(z))}{\opt_k(S)} ~\le~
			\frac{16\alpha\cdot d^2_H(s,s'_i)}{\cost_S(\{ s'_1,\dots,s'_{k'} \})} ~\le~
			\frac{16\alpha\cdot d^2_H(s,s'_i)}{\cost_{S'_i}(s'_i)} ~\le~
			\tilde\sigma(f_s).
	\]
	Second, if $z$ is such that for some $j\in [k]$ we have $d^2_H(s,s_j(z)) \le 16\cdot R_i$,
	then we use again that $\cost_S(\{ s'_1,\dots,s'_{k'} \}) \le \alpha\cdot \opt_k(S)$
	to obtain
	\begin{align*}
		\frac{f_s(z)}{\displaystyle \sum_{s'\in S} f_{s'}(z)} ~&\le~ 
			\frac{d^2_H(s,s_j(z))}{\opt_k(S)} ~\le~
			\frac{16\alpha\cdot R_i}{\cost_S(\{ s'_1,\dots,s'_{k'} \})} ~\le~
			\frac{16\alpha\cdot R_i}{\cost_{S'_i}(s'_i)} ~=~
			\frac{16\alpha\cdot R_i}{|S'_i|\cdot R_i} ~=~
			\frac{16\alpha}{|S'_i|}\\~&\le~
			\tilde\sigma(f_s).
	\end{align*}
	It remains to handle the case when, for the $z\in \RR^{4k}$ under consideration,
	we have $d^2_H(s,s_j(z)) > 16\cdot d^2_H(s,s'_i)$ and 
	$d^2_H(s,s_j(z)) > 16\cdot R_i$ for all $j\in [k]$.
	Let $s''_j$ be the segment defined by $z$ that is closest to $s$, 
	which means that $f_s(z)=d^2_H(s,s''_j)$.
	We then have $d_H(s,s''_j) > 4\cdot d_H(s,s'_i)$ and $d_H(s,s''_j) > 4\cdot \sqrt{R_i}$.
	Because of the triangular inequality for $d_H(\cdot,\cdot)$ and because $s''_j$
	is closest to $s$, for each segment $s'\in S$
	and each $\ell\in [k]$ we have
	\[
		d_H(s,s''_j) ~\le~ d_H(s,s_\ell(z)) ~\le~ 
			d_H(s,s'_i)+ d_H(s'_i,s') + d_H(s',s_\ell(z)).
	\]
	This means that,
	\[
		\forall s'\in S:~~~ \min\{ d_H(s',s_1(z)),\dots, d_H(s',s_k(z))\}  
			~\ge~  d_H(s,s''_j) - d_H(s,s'_i)- d_H(s'_i,s').
	\]
	Using that $d_H(s,s'_i) < d_H(s,s''_j) /4$ we get:
	\[
		\forall s'\in S:~~~ \min\{ d_H(s',s_1(z)),\dots, d_H(s',s_k(z))\} 
			~\ge~  \tfrac 34 d_H(s,s''_j) - d_H(s'_i,s'). 
	\]
	At least half of the segments $s'\in S'_i$ must have $d^2_H(s',s'_i)\le 2R_i$;
	otherwise the average cost of the elements of $S'_i$ would be larger than $R_i$.
	Let $S''_i$ be those segments $s'$ of $S'_i$ with $d^2_H(s',s'_i)\le 2R_i$.
	We then have
	\begin{align*}
		\forall s'\in S''_i:~~~ \min\{ d_H(s',s_1(z)),\dots, d_H(s',s_k(z))\} ~&\ge~  
					\tfrac 34 d_H(s,s''_j) - d_H(s'_i,s')\\ ~&\ge~
					\tfrac 34 d_H(s,s''_j) - \sqrt{2 R_i}\\ ~&\ge~
					\tfrac 34 d_H(s,s''_j) - \tfrac{1}{2\sqrt{2}} d_H(s,s''_j) \\ ~&\geq~
					\tfrac 14 d_H(s,s''_j).
	\end{align*}
	and therefore
	\begin{align*}
		\forall s'\in S''_i:~~~ f_{s'}(z) = \bigl(\min\{ d_H(s',s_1(z)),\dots, d_H(s',s_k(z))\}\bigr)^2 ~&\ge~  \tfrac{1}{16} d^2_H(s,s''_j).
	\end{align*}
	It follows that
	\[
		\frac{f_s(z)}{\displaystyle \sum_{s'\in S} f_{s'}(z)} ~\le~ 
			\frac{d^2_H(s,s''_j)}{\displaystyle \sum_{s'\in S''_i} f_{s'}(z)} ~\le~
			\frac{d^2_H(s,s''_j)}{ |S''_i| \cdot \tfrac{1}{16} d^2_H(s,s''_j)} ~\leq~
			\frac{2\cdot 16}{|S'_i|} ~\le~ \tilde\sigma(f_s) .
	\]

	We have shown that
	\[
		\forall z\in \RR^{4k}:~~~ \frac{f_s(z)}{\displaystyle \sum_{s'\in S} f_{s'}(z)} ~\le~ 
			\tilde\sigma(f_s).
	\]
	It then follows that
	\[
		\sigma(f_s,F) ~=~ \sup_{z\in \RR^{4k}} \frac{f_s(z)}{\displaystyle \sum_{f_{s'}\in F} f_{s'}(z)} 
			~\le~ \tilde\sigma(f_s).\qedhere
	\]	
\end{proof}

Finally, note that for the sensitivities $\tilde\sigma(f_s)$ defined in Lemma~\ref{le:sensitivity}, we have
the total sensitivity
\begin{align*}
	\tilde\Sigma(F) ~~&=~~ \sum_{s\in S} \tilde\sigma(f_s) 
					~~=~~ \sum_{s\in S} \left(\frac{32\alpha}{|S'_{\iota(s)}|} + \frac{16\alpha\cdot d^2_H(s,s'_{\iota(s)})}{\cost_{S'_{\iota(s)}}(s'_{\iota(s)})}\right)\\
					~~&=~~  \sum_{i\in [k']}\left( \sum_{s\in S'_i} \frac{32\alpha}{|S'_i|} + \sum_{s\in S'_i} \frac{16\alpha\cdot d^2_H(s,s'_i)}{\cost_{S'_i}(s'_i)}\right)
					~~=~~  \sum_{i\in [k']}\left( 32\alpha + 16\alpha \right)\\
					~~&=~~ O(\beta \alpha k).
\end{align*}

Next, we bound the VC-dimension of the range space associated with the input segments.

\begin{lemma}
\label{le:VC}
	Assume that we have a weight $w_s > 0$ for each $s\in S$ and consider
	the set of functions $F_w=\{ w_s\cdot f_s\mid s\in S \}$.
	The range space $\range(F_w)$ has VC-dimension $O(1)$.
\end{lemma}
\begin{proof}
    First note that the range space $\range(F_w)$ is equivalent to the range space $(S, R)$, where
    the ranges are
    \begin{align*}
		R~&=~\bigl\{ \{ s\in S \mid (w_s\cdot f_s)(z)\ge r \} \mid z\in \RR^{4k},~ r \in [0,\infty) \bigr\}\\
		&=~ \bigl\{ \{ s\in S \mid w_s\cdot \min\{ d^2_H(s,s_1(z))\dots, d^2_H(s,s_k(z))\} \ge r \} \mid z\in \RR^{4k},~ r \in [0,\infty) \bigr\}\\
			~&=~\bigl\{ \{ s\in S \mid \forall i\in [k]:~\sqrt{w_s} \cdot d_H(s,s_i(z)) \ge \sqrt{r} \} \mid z\in \RR^{4k},~ r \in [0,\infty) \bigr\}.
	\end{align*}
	Setting $w'_s=\sqrt{w_s}$ for each $s\in S$ and $r'=\sqrt{r}$, we get that
	the ranges are
	\[
		R ~=~\bigl\{ \{ s\in S \mid \forall i\in [k]:~ w'_s \cdot d_H(s,s_i(z)) \ge r' \} \mid z\in \RR^{4k},~ r' \in [0,\infty) \bigr\}.
	\]
	For each segment $s\in S$, consider the hypersurface $\lambda_s$ in $\RR^{4k+1}$
	given by the graph of the function 
	$z\in \RR^{4k}\mapsto w'_s \cdot d_H(s,s_i(z))$. This is 
	$
		\lambda_s ~\equiv~ \left\{ \bigl( z, w'_s \cdot d_H(s,s_i(z)) \bigr)\in \RR^{4k}\times \RR\mid z \in \RR^{4k} \right\}$.	
	As it has been discussed and used in Section~\ref{sec:algebraic} when defining the set $\FF(ab,a'b')$,
	the hypersurface $\lambda_s$ is contained in the union of $8$ \emph{algebraic} hypersurfaces of bounded degree,
	each of them being the graph of a function.
	Let $\Lambda_s$ be the set of those $8$ algebraic hypersurfaces for the segment $s\in S$.

	Set $\Lambda:= \cup_{s\in S} \Lambda_s$ and let $\AA$ be the arrangement in $\RR^{4k+1}$ induced
	by $\Lambda$.
	Each point $(z,r')\in \RR^{4k}\times \RR$ gives a range to $R$,
	and two points in the same cell of $\AA$ give exactly the same range to $R$ because, for each $s\in S$,
	the surface $\lambda_s$ is above, below or on all the points of the cell.
	It may happen that points in different cells of $\AA$ give the same range, as one still has
	to check the condition $\forall i\in [k]:~ w'_s \cdot d_H(s,s_i(z)) \ge r'$.
	In any case, the number of cells in $\AA$ is an upper bound to the number of ranges in $R$,
	which is exactly the number of ranges in $\range(F_w)$.
	
	Classical results in Real Algebraic Geometry imply that $\AA$ has $|\Lambda|^{O(k)}$ cells;
	see for example~\cite[Chapter 7]{AlgRealAlgGeom06} or~\cite[Section 6.2]{Matousek02}.
	This implies that the so-called shattering dimension of $\range(F_w)$ is $O(k)=O(1)$.
	(See for example Har-Peled~\cite[Chapter 5]{HarPeled11} for the concept and the next property.)
	Since a range space has bounded shattering dimension if and only if it has bounded VC-dimension
	this implies that the VC-dimension of $\range(F_w)$ is $O(1)$.
	The approach we have used is essentially an application of the methodology
	discussed by Matou{\v s}ek \cite[Section 10.3]{Matousek02}.
	
	Note that in this proof we have not tried to optimize the bound on the VC-dimension
	because we assume $k$ is constant.
	Perhaps a better bound follows from adapting the result of Driemel at al.~\cite{Driemel21} 
	to the case of weights.
\end{proof}

We can now apply Theorem~\ref{CoresetThm} on $F$ to obtain the coreset.

\begin{theorem}
\label{thm:coresets}
	Assume that $k$ is a fixed positive integer. 
	Let $\delta,\eps$ be real values in $(0,1/2)$.
	For any set $S$ of $n$ unweighted segments in the plane, we can compute in time $O(n\log(1/\delta))$ a 
	subset $T\subseteq S$ of 
	\[
		O\left( \eps^{-2} \log \frac{1}{\delta} \right)
	\]
	segments and weights $u_s> 0$ for each $s\in T$ such that, 
	with probability at least $1-\delta$:
	\[
		\forall \text{ segments }s_1,\dots,s_k:~~~ \left| \cost_S(\{ s_1,\dots,s_k\})  - \cost_T(\{ s_1,\dots,s_k\})\right| ~\le~
		\eps \cdot \cost_S(\{ s_1,\dots,s_k\}).
	\]
\end{theorem}

\begin{proof}
We first compute a bicriteria $(\alpha=O(1),\beta=O(1))$-approximation for $k$-means on $S$ by using the algorithm of Chen~\cite[Theorem A.4]{Chen09}, which in turn is a modification of the algorithm by Indyk~\cite{Indyk99}.
For a probability of error $\delta'=\delta/2$, the algorithm takes $O(n\log(1/\delta'))=O(n\log(1/\delta))$ time
and succeeds with probability at least $1-\delta'$ in finding a set $s_1,\dots,s_{k'}$ of $k'=O(k)$ segments such that
$\cost_S(\{ s_1,\dots,s_{k'}\})\le O(1)\cdot \cost_S(\{ s_1,\dots,s_k\})$.
Note that the algorithm of Chen is for the discrete version of $k$-means, where the centers under consideration must be a subset of $S$. However, it is well-known that the triangle inequality implies that this is a factor $4$ off for the continuous $k$-means version. This factor $4$ is then subsumed by the $O(1)$ approximation factor.

Let $F=\{f_s \mid s\in S\}$. 
We use the bicriteria approximation for the sensitivity upper bounds $\tilde\sigma(f_s)$, for each $f_s\in F$, as defined in Lemma~\ref{le:sensitivity}. As discussed after Lemma~\ref{le:sensitivity}, the total sensitivity $\tilde\Sigma(F)$ is $O(1)$ and, by Lemma~\ref{le:VC}, the VC-dimension of $\range(F)$ is $O(1)$. The result then follows using Theorem~\ref{CoresetThm} with probability of error $\delta'=\delta/2$. The size of the set $C\subseteq F$ selected by Theorem~\ref{CoresetThm} is $O(\tilde\Sigma(F)\eps^{-2} (\log \tilde\Sigma(F) + \log(1/\delta')) = O(\eps^{-2} \log(1/\delta))$, and each function $f\in C$ has a given weight $u_f>0$. 
We set $T = \{ s\in S\mid f_s\in C \}$ and, for each segment $s\in T$, we define the weight $w_s:=u_{f_s}$.

With probability at least $1-\delta'$ we have 
\[
		\forall z\in\RR^{4k}:~~~ \left| \sum_{f\in F} f(z) - \sum_{f\in C}u_f\cdot f(z)\right| ~\le~
		\eps \sum_{f\in F} f(z),
\]
which can be rewritten as
\[
		\forall z\in\RR^{4k}:~~~ \left| \sum_{s\in S} f_s(z) - \sum_{s\in T}w_s\cdot f_s(z)\right| ~\le~
		\eps \sum_{s\in S} f_s(z).
\]
Since $f_s(z)= \min \{ d^2_H(s,s_1(z)),\dots, d^2_H(s,s_k(z))\}$ and $s_1(z),\dots,s_k(z)$ goes over all $k$ tuples of candidate segments when $z$ iterates over all $\RR^{4k}$, the last statement is equivalent to
\[
	\forall \text{ segments }s_1,\dots,s_k: ~~~ \left| \cost_S(\{ s_1,\dots,s_k\})  - \cost_T(\{ s_1,\dots,s_k\})\right| ~\le~
		\eps \cdot \cost_S(\{ s_1,\dots,s_k\}).
\]
The algorithm may fail only if the bicriteria approximation of Chen fails or if the application of Theorem~\ref{CoresetThm} fails, and each of them separately fails with probability at most $\delta/2$.
\end{proof}


\section{Putting it all together}

Let $S$ be a set of $n$ segments in the plane without weights.
We first set a fixed probability of error $\delta=1/2$, which means that the terms $\log(1/\delta)$ become $O(1)$.
We keep using $\eps\in (0, 1/2)$ as a parameter.

We first compute a weighted coreset $T\subseteq S$ with $|T| = O(\eps^{-2}))$ elements in $O(n)$ time as described in Theorem~\ref{thm:coresets}; for each segment $s\in T$ we have a weight $w_s>0$. 
If $S^* = \{s^*_1,\ldots,s^*_k\}$ is an optimal set of segments for $S$, then from Theorem~\ref{thm:coresets} we have that
$\cost_T(S^*)\leq (1+\eps) \cdot \cost_S(S^*)$ with probability at least $1/2$.

We apply the $(1+\eps)$-approximation algorithm of Theorem~\ref{thm:main_via_Vigneron} on $T$, taking into account the weights of the segments.
As $|T| = O(\eps^{-2})$ the algorithm runs in time
\[
	O\left( (\eps^{-2})^{8k -2+\eta} + (\eps^{-3})^{4k+1} \log^{4k+1} (\eps^{-3})\right) ~=~
	O\left(\eps^{-16k+4-\eta} + \eps^{-12k-3} \log^{4k+1} (\eps^{-1})\right)
\]
for any $\eta>0$.
When $k=1$, the second summand dominates.
When $k\ge 2$ and $\eps$ is below some constant $\eps_0$, the first summand dominates.

Let $T^* = \{t^*_1,\ldots,t^*_k\}$ be an optimal set of segments for $k$-means of the weighted set $T$.
The algorithm of Theorem~\ref{thm:main_via_Vigneron} has then provided a set $S_\eps = \{ s_{1,\eps},\dots,s_{k,\eps}\}$
of $k$ segments for which $\cost_T(S_\eps) \leq (1+\eps) \cdot \cost_T(T^*)$.
Note that for the set $S_\eps$ we also get from Theorem~\ref{thm:coresets} that $(1-\eps)\cdot \cost_S(S_\eps) \leq \cost_T(S_\eps)$. Since $\cost_T(T^*) \leq \cost_T(S^*)$, we conclude that 
\begin{align*}
(1-\eps)\cdot \cost_S(S_\eps) ~&\leq~ \cost_T(S_\eps) ~\leq~ (1+\eps) \cdot \cost_T(T^*) ~\leq~ (1+\eps) \cdot \cost_T(S^*)\\
~&\leq~ (1+\eps)^2 \cdot \cost_S(S^*)
\end{align*}
or 
\[
\cost_S(S_\eps) ~\leq~ \frac{(1+\eps)^2}{(1-\eps)} \cdot \cost_S(S^*) ~=~  (1+O(\eps)) \cdot \cost_S(S^*).
\]
Setting $\eps=\Theta(\eps')$ appropriately, we get a $(1+\eps')$-approximation for any desired $\eps'$.

By independently repeating the algorithm $O(\log (1/\delta))$ times and taking the best among the solutions,
we can reduce the probability of error to any given value $\delta$. Because $k=O(1)$, evaluating each candidate solution with respect to the whole set of segments takes $O(n)$ time. We summarize in the following.

\begin{theorem}
	Let $k$ a fixed, positive integer and let $\delta,\eps\in (0,1/2)$.
	Let $S$ be a family of $n$ unweighted segments in the plane.
	We can compute $k$ segments $s_{1,\eps},\dots,s_{k,\eps}$ in $\RR^2$ such that, with probability at least $1-\delta$,
\[
\cost_S(\{s_{1,\eps},\dots,s_{k,\eps}\}) ~\le~ (1+\eps) \min_{s_1,\dots,s_k} \cost_S(\{ s_1,\dots, s_k\})
\]
in time $O\left(\left(n+ \eps^{-16k+4-\eta} + \eps^{-12k-3} \log^{4k+1} (\eps^{-1})\right) (\log(1/\delta)\right)$, for any $\eta>0$.
\end{theorem}
 
For $k=1$, the running time is $O\left(\left(n+ \eps^{-15} \log^{5} (\eps^{-1})\right) (\log(1/\delta)\right)$,
while for $k\ge 2$ the running time is $O\left(\left(n+ \eps^{-16k+4-\eta}\right) (\log(1/\delta)\right)$ for any $\eta>0$.

\section{Extension to polylines}
In this section we briefly discuss the extension of our result to the case of polylines of bounded complexity.
To reduce the number of parameters, we assume that each polyline has at most $\ell$ segments and we search the $k$-means among polylines that have at most $\ell$ segments. (We can also handle the case where the input and the target centers have different complexities.) To simplify the discussion, we assume that each input polyline has exactly $\ell$ segments. We further assume that $\ell=O(1)$.

We regard each polyline $\pi$ as the union of segments and note that
the distance between the polyline $\pi$ with segments $s_1,\dots, s_\ell$ and the polyline $\pi'$ 
with segments $s'_1,\dots, s'_\ell$ is 
\[
	d_H(\pi,\pi')~=~ \max\Bigl\{  \max_{i\in [\ell]}\min_{j\in [\ell]} d_H(s_i,s'_j),~ 
									\max_{j\in [\ell]}\min_{i\in [\ell]} d_H(s'_j,s_i) \Bigr\}.
\]
Therefore the distance between any two polylines is described as a max-min combination of $O(\ell^2)=O(1)$ values.

A polyline with $\ell$ segments is parameterized by $2(\ell+1)$ real values.
Therefore, a sequence of $k$ polylines with $\ell$ segments each is parameterized by a point in $\RR^\kappa$ for $\kappa=2k(\ell+1)$. (Before, for segments, we had $\kappa=4k$.) Each $z\in \RR^\kappa$ defines $k$ polylines 
$\pi_1(z),\dots,\pi_k(z)$, each consisting of $\ell$ segments.

Let $\Pi$ be a set of polylines in the plane, each with $\ell$ segments.
For each $\pi\in \Pi$, we define the function $f_\pi: \RR^\kappa \rightarrow \RR$ by
\[
	f_\pi(z) ~:=~ \min \{ d^2_H(\pi,\pi_1(z)),\dots, d^2_H(\pi,\pi_k(z))\} ~=~ 
				\bigl( \min \{ d_H(\pi,\pi_1(z)),\dots, d_H(\pi,\pi_k(z))\}\bigr)^2 
\]
and then define the set of functions $F=\{f_\pi\mid \pi\in \Pi\}$.

We first note that the VC-dimension of the range space $\range(F_w)$ is $O(1)$, where $F_w$ is obtained from $F$ by scaling each $f_\pi\in F$ with a different scalar $w_\pi>0$. 
The proof of Lemma~\ref{le:VC} readily applies to this case as it only relies on the description complexity of $d_H(\pi,\pi_i(z))$ being constant, and each patch of the description being an algebraic function. 

Next we note that we can use the bicriteria $(\alpha=O(1),\beta=O(1))$-approximation for $k$-means of Chen, as we did in the proof of Theorem~\ref{thm:coresets}. Indeed, this algorithm only requires that we can compute the distance between any two input objects, which we can do in constant time. The rest of the proof of Theorem~\ref{thm:coresets} goes unchanged because Lemma~\ref{le:sensitivity} and Theorem~\ref{CoresetThm} do not make any assumption related to segments beyond the VC-dimension. We thus obtain with probability at least $1/2$ a coreset $\tilde\Pi$ of $O(\eps^{-2})$ input polylines, each of them with a positive weight $w_\pi$.

It remains to adapt Theorem~\ref{thm:main_via_Vigneron} to the setting of polylines.
As we have done in Theorem~\ref{thm:nice_family}, for each polyline $\pi$ we can compute a family $\FF_\pi$ of nice functions such that 
\[
	f_\pi(z) ~=~ \sum_{f\in\FF_\pi}f(z) ~=~ \min_{i\in [k]} d^2_H(\pi,\pi_i(z)) ~~\text{ for all $z\in \RR^\kappa$.}
\]
Indeed, as we did in the proof of Theorem~\ref{thm:nice_family}, we can break the parameter space $\RR^\kappa$ using $O(k^2\ell^2)=O(1)$ algebraic hypersurfaces into $O(1)$ cells such that, within each cell, the max-max-min expression defining $d_H(\pi,\pi_i(z))$ is always the same algebraic expression. 
We can then apply Theorem~\ref{thm:Vigneron} to the family of nice functions $\cup_{\pi\in\tilde\Pi} \FF_\pi$, where each function in $\FF_\pi$ has been scales with the corresponding weight $w_\pi$.
Thus, we have an application of Theorem~\ref{thm:Vigneron} in $\RR^\kappa$ for $O(\eps^{-2})$ functions. The running time is, for any $\eta>0$,
\[ 
	O((\eps^{-2})^{2\kappa-2+\eta} + (\eps^{-3})^{\kappa+1} \log^{\kappa+1}(\eps^{-3})) ~=~ 
	O(\eps^{-4\kappa+4+\eta} + \eps^{-3\kappa-3} \log^{\kappa+1}(\eps^-1)) ~=~
	O(\eps^{-O(k\ell)}).
\]	
Like before, we can make $O(\log (1/\delta))$ independent repetitions to decrease the probability of failure to $\delta$.
We summarize below.

\begin{theorem}
	Let $k$ and $\ell$ be fixed, positive integers and let $\delta,\eps\in (0,1/2)$ be parameters.
	Let $\Pi$ be a family of $n$ unweighted polylines in the plane, each with at most $\ell$ segments.
	We can compute $k$ polylines $\pi_{1,\eps},\dots,\pi_{k,\eps}$ in $\RR^2$, each with at most $\ell$ segments,
	such that, with probability at least $1-\delta$,
\[
\cost_\Pi(\{\pi_{1,\eps},\dots,\pi_{k,\eps}\}) ~\le~ (1+\eps) \min_{\pi_1,\dots,\pi_k} \cost_\Pi(\{ \pi_1,\dots, \pi_k\})
\]
in time $O\left(\left(n+ \eps^{-O(k\ell)}\right) (\log(1/\delta)\right)$.
\end{theorem}

\bibliography{main}

\end{document}